\documentclass[12pt,draftcls,onecolumn]{IEEEtran}

\IEEEoverridecommandlockouts

\usepackage{latexsym}
\usepackage{graphicx}
\usepackage{epsfig}
\usepackage{subfigure}
\usepackage{array}
\usepackage{amsmath}
\usepackage{amssymb}
\usepackage{color, soul}
\usepackage{amsthm}
\usepackage{enumerate}
\usepackage{algpseudocode}
\usepackage{algorithm}

\theoremstyle{plain} 
\newtheorem{thm}{Theorem}

\newtheorem{lemma}{Lemma}
\newtheorem{prop}{Proposition}
\theoremstyle{definition}

\newtheorem{remark}{Remark}

\newcommand{\bxi} {\boldsymbol{\xi}}
\newcommand{\blam}{{\boldsymbol \lambda}}
\newcommand{\bW}{{\bf W}}
\newcommand{\bH}{{\bf H}}
\newcommand{\bR}{{\bf R}}
\newcommand{\bK}{{\bf K}}
\newcommand{\bI}{{\bf I}}

\newcommand{\bM}{{\bf M}}
\newcommand{\bA}{{\bf A}}
\newcommand{\bB}{{\bf B}}
\newcommand{\bC}{{\bf C}}
\newcommand{\bD}{{\bf D}}
\newcommand{\bQ}{{\bf Q}}
\newcommand{\bZ}{{\bf Z}}
\newcommand{\bS}{{\bf S}}
\newcommand{\bT}{{\bf T}}
\newcommand{\bX}{{\bf X}}

\newcommand{\bLam}{{\bf \Lambda}}

\newcommand{\by}{{\bf y}}
\newcommand{\bx}{{\bf x}}

\newcommand{\bdz}{\Delta{\bf z}}
\newcommand{\bdw}{\Delta{\bf w}}
\newcommand{\bdlam}{\Delta \blam}
\newcommand{\ba}{{\bf a}}
\newcommand{\bb}{{\bf b}}

\newcommand{\br}{{\bf r}}
\newcommand{\bz}{{\bf z}}
\newcommand{\bo}{{\bf 0}}
\newcommand{\bk}{{\bf k}}
\newcommand{\bw}{{\bf w}}

\newcommand{\sS}{\mathcal{S}}
\newcommand{\bp} {\begin{proof}}
\newcommand{\ep} {\end{proof}}

\newcommand{{\Rb}} {\right)}

\newcommand{{\Rf}} {\right\}}

\def\bal#1\eal{\begin{align}#1\end{align}}

\begin{document}

\title{An Algorithm for Global Maximization of Secrecy Rates in Gaussian MIMO  Wiretap Channels}

\author{Sergey Loyka, Charalambos D. Charalambous
%
\thanks{S. Loyka is with the School of Electrical Engineering and Computer Science, University of Ottawa, Ontario,
Canada, K1N 6N5, e-mail: sergey.loyka@ieee.org.}

\thanks{C.D. Charalambous is with the ECE Department, University of Cyprus, 75 Kallipoleos Avenue, P.O. Box 20537, Nicosia, 1678, Cyprus, e-mail: chadcha@ucy.ac.cy}
}
%
\maketitle
%

\begin{abstract}
Optimal signaling for secrecy rate maximization in Gaussian MIMO wiretap channels is considered. While this channel has attracted a significant attention recently and a number of results have been obtained, including the proof of the optimality of Gaussian signalling, an optimal transmit covariance matrix is known for some special cases only and the general case remains an open problem. An iterative custom-made algorithm to find a globally-optimal transmit covariance matrix in the general case is developed in this paper, with guaranteed convergence to a \textit{global} optimum. While the original optimization problem is not convex and hence difficult to solve, its minimax reformulation can be solved via the convex optimization tools, which is exploited here. The proposed algorithm is based on the barrier method extended to deal with a minimax problem at hand. Its convergence to a global optimum is proved for the general case (degraded or not) and a bound for the optimality gap is given for each step of the barrier method. The performance of the algorithm is demonstrated via numerical examples. In particular, 20 to 40 Newton steps are already sufficient to solve the sufficient optimality conditions with very high precision (up to the machine precision level), even for large systems. Even fewer steps are required if the secrecy capacity is the only quantity of interest. The algorithm can be significantly simplified for the degraded channel case and can also be adopted to include the per-antenna power constraints (instead or in addition to the total power constraint). It also solves the dual problem of minimizing the total power subject to the secrecy rate constraint.
\end{abstract}

%
\section{Introduction}
\label{sec:introduction}

Wide-spread use of wireless systems has initiated significant interest in their security and related information-theoretic studies \cite{Bloch}. Secrecy capacity has emerged as a key performance metric, which extends the regular channel capacity to accommodate the secrecy requirement. Wyner's wire-tap channel (WTC) \cite{Bloch}-\cite{Massey} is the most popular model to accommodate secrecy, which was extended to the Gaussian channel \cite{Cheong} and subsequently to the Gaussian multiple-input multiple-output (MIMO) setting \cite{Khisti-1}-\cite{Liu}; the reader is referred to \cite{Bloch} for a detailed discussion of this model and extensive literature review. The Gaussian MIMO WTC has been recently a subject of intense study and a number of results have been obtained, including the proof of optimality of Gaussian signaling \cite{Bloch}, \cite{Khisti-1}-\cite{Liu}. While the functional form of the optimal (capacity-achieving) distribution has been established, significantly less is known about its optimal covariance matrix (the only remaining parameter to completely characterize the distribution since the mean is always zero).

The optimal transmit covariance matrix under the total power constraint has been obtained for some special cases, e.g. low/high SNR, multiple-input single-output (MISO) channels, full-rank, rank-1 or weak eavesdropper cases, or the parallel channel \cite{Khisti-1}-\cite{Gursoy}, but the general case remains illusive. The main difficulty lies in the fact that the underlying optimization problem is in general not a convex problem. It was conjectured in \cite{Oggier} and proved in \cite{Khisti-2} using an indirect approach (via the degraded channel) that the optimal signaling is on the positive directions of the difference channel (where the legitimate channel is stronger than the eavesdropper one).  A direct proof based on the necessary Karush-Kuhn-Tucker (KKT) optimality conditions has been obtained in \cite{Loyka}. A weaker form of this result (non-negative instead of positive directions) has been obtained earlier in \cite{Li}. In the general case, the rank of an optimal covariance matrix does not exceed the number of positive eigenvalues of the difference channel matrix \cite{Loyka}. An exact full-rank solution for the optimal covariance has been obtained in \cite{Loyka} and its properties have been characterized. In particular, unlike the regular channel (no eavesdropper), the optimal power allocation does not converge to uniform one at high SNR and the latter remains sub-optimal at any finite SNR. In the case of weak eavesdropper (its singular values are much smaller than those of the legitimate channel), the optimal signaling mimics the conventional one (water-filling over the channel eigenmodes) with an adjustment for the eavesdropper channel. The rank-one solution in combination with the full-rank one provides a complete solution for the case of two transmit antennas and any number of receive/eavesdropper antennas. The 2-2-1 case (2 transmit, 2 receive, 1 eavesdropper antenna) has been studied earlier in \cite{Shafiee-09} and the MISO case (single-antenna receiver) has been considered in \cite{Li-07}\cite{Shafiee-07} and settled in \cite{Khisti-07}\cite{Khisti-1}, for which beamforming is optimal and which is also the case for a MIMO-WTC in the low SNR regime. The case of isotropic eavesdropper is studied in detail in \cite{Loyka-13}, including the optimal signaling in an explicit closed form and its properties. This case is shown to be the worst-case MIMO wire-tap channel. Based on this, lower and upper (tight) capacity bounds have been obtained for the general case, which are achievable by an isotropic eavesdropper. The set of channels for which isotropic signaling is optimal has been fully characterized \cite{Loyka-13}. It turns out to be more richer than that of the conventional (no eavesdropper) MIMO channel. A closed-form solution was obtained in \cite{Loyka-14} for the case of weak eavesdropper but otherwise arbitrary channel; its optimal power allocation somewhat resembles the water-filling but is not identical to it. For the case of parallel channels, independent signaling is optimal \cite{Khisti-08}\cite{Li-2}, which implies that the optimal covariance matrix is diagonal; the corresponding optimal power allocation can be found in \cite{Li-2}. This also implies that the eigenvectors of optimal covariance matrix are the same as the right singular vectors of the legitimate or eavesdropper channels when the latter two are the same \cite{Loyka-14} and the corresponding power allocation is the same as in \cite{Li-2}. The low-SNR regime has been studied in detail in \cite{Gursoy}. In particular, signaling on the strongest eigenmode(s) of the difference channel matrix is optimal. Little is known beyond these special cases and the general case is still an open problem.

While numerical algorithms have been proposed in \cite{Li-13,Steinwandt-14} to compute a transmit covariance matrix for the MIMO-WTC, their convergence to a \textit{global} optimum has not been proved. The main difficulty lies in the fact that the underlying optimization problems are not convex and hence KKT conditions are not sufficient for optimality \cite{Bertsekas}. In particular, while the alternating optimization algorithm in \cite{Li-13} is shown to convergence to a KKT (stationary) point, it is not necessarily a global maximum (due to the above reason); it may, in fact, be a saddle point or a \textit{local} rather than \textit{global} maximum of the secrecy rate\footnote{For non-convex problems, KKT point can also be a local minimum rather than maximum. This is ruled out in \cite{Li-13} by the non-decreasing nature of the generated sequence of objective values.} and it is not known how far away it is from the global maximum. This remark also applies to the algorithms considered in \cite{Steinwandt-14,Alvarado}.

The purpose of this paper is to develop a numerical algorithm for computing a \textit{globally}-optimal covariance matrix in the general case, i.e. for the general Gaussian MIMO-WTC (degraded or not), with guaranteed convergence to a \textit{global} optimum, and to prove its convergence. This is a challenging task as the underlying optimization problem is not convex so that standard tools of convex optimization cannot be used; in general, non-convex problems are much harder to solve \cite{Boyd}. We deal with this challenge by using the minimax representation of the secrecy capacity found in \cite{Khisti-2}. While this representation appears to be more complicated than the standard one (the former involves two conflicting optimizations while the latter - only one), it turns out to be much easier to solve, at least numerically, as we demonstrate using the primal-dual representation of Newton method in combination with the barrier method. The main advantage of this approach is that each of the two problems is convex, the saddle-point property holds and hence the respective KKT conditions are sufficient for global optimality (Slater's condition holds as well). A conceptually-similar approach has been used before for optimizing the transmitter with per-antenna power constraints in the regular (no secrecy) MIMO broadcast channel in \cite{Yu}. Our custom-made algorithm essentially solves the KKT optimality conditions (see e.g. \cite{Boyd} for a background on these conditions), which are sufficient for the minimax problem at hand, in an iterative way using the primal-dual representation of Newton method in combination with the barrier method (to accommodate inequality constraints) adopted to the MIMO WTC setting, see Section V.  A proof of the algorithm's convergence to a \textit{global} optimum is also provided for the general case. While we formulate the algorithm for the total power constraint, it can be easily modified to accommodate other forms of power constraint, e.g. maximum per-antenna constraint (instead or in addition to the total power constraint), and also to solve a dual problem of minimizing the total transmit power under the secrecy rate constraint.

A key part of the convergence proof for our algorithm involves a proof of non-singularity of the KKT matrix\footnote{A singular KKT matrix would imply that the corresponding Newton step is not defined and thus the algorithm would terminate without converging to a global optimum.}, so that Newton steps are well-defined for all iterations of the algorithms and they generate a sequence of norm-decreasing residuals and hence converge to a globally-optimal point (i.e. a solution of the KKT conditions which corresponds to zero residual). This is a difficult task since the underlining optimization problems involve both maximization and minimization and the corresponding KKT matrix is indefinite so that the regular tools developed for positive semi-definite matrices \cite{Horn-1} do not apply. A block-partitioned factorization of the KKT matrix is used to accomplish it. This is explained in Section V, which also gives a bound on the optimality gap for each step of the barrier method. Numerical examples in Section VII demonstrate fast convergence of the algorithm: 20 to 40 Newton steps are already sufficient to achieve a very high precision (up to the machine precision level), even for large system. Even less steps are required if the secrecy capacity is the only quantity of interest. Section VI demonstrates that significant simplifications in the algorithm are possible for a degraded channel.  Section IV gives a brief review of the barrier and Newton methods for inequality-constrained optimization, and presents an algorithm for minimax problems with guaranteed convergence to a global optimum. Section III summarizes the minimax representation of the secrecy capacity on which our algorithm is based. Section 2 reviews the Gaussian  MIMO-WTC model and its secrecy capacity.

\section{Wire-Tap Gaussian MIMO Channel Model}
Let us consider the standard Gaussian MIMO wire-tap channel model,
\bal
\label{eq1}
\by_1 = \bH_1\bx+\bxi_1, \quad \by_2 = \bH_2\bx+\bxi_2
\eal
where $\bx=[x_1 ,x_2 ,...x_m ]'\in R^{m,1}$ is the (real) transmitted  signal vector of dimension $m \times 1$, $'$ denotes transposition, $\by_{1(2)} \in R^{n_{1(2)},1}$ are the (real) received vectors at the receiver (eavesdropper), $\bxi_{1(2)} $ is the additive white Gaussian noise at the receiver (eavesdropper) (normalized to unit variance in each dimension), $\bH_{1(2)} \in R^{n_{1(2)} ,m}$ is the $n_{1(2)} \times m$ matrix of the channel gains between each Tx and each receive (eavesdropper) antenna, $n_{1(2)} $ and $m$ are the number of Rx (eavesdropper) and Tx antennas respectively. The channels $\bH_{1(2)} $ are assumed to be quasistatic (i.e., constant for a sufficiently long period of time so that the infinite horizon information theory assumption holds) and frequency-flat, with full channel state information (CSI) at the Rx and Tx ends. A secrecy rate is achievable for this channel if (i) the receiver is able to recover the message with arbitrary low error probability (reliability criterion) and (ii) the information leaked to the eavesdropper approaches zero asymptotically (secrecy criterion) \cite{Bloch}.

\begin{figure}[t]
\label{fig.1a}
\centerline{\includegraphics[width=2.5in]{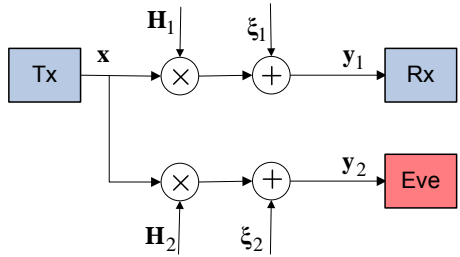}}
\caption{ A block diagram of the Gaussian MIMO wiretap channel. Full channel state information is available at the transmitter. $\bH_{1(2)}$ is the channel matrix to the legitimate receiver (eavesdropper); $\bx$ is the transmitted signal and $\by_{1(2)}$ is the received (eavesdropper) signal; $\bxi_{1(2)} $ is the AWGN at the receiver (eavesdropper). The information leakage to the eavesdropper is required to approach zero asymptotically.}
\end{figure}

For a given transmit covariance matrix $\bR=E\{\bx\bx'\}$, where $E\{\cdot\}$ is statistical expectation, the maximum achievable secrecy rate between the Tx and Rx (so that the rate between the Tx and eavesdropper is zero) is \cite{Khisti-2}-\cite{Liu}
\begin{equation}
\label{eq.C(R)}
C(\bR)= \frac{1}{2} \ln \frac{|\bI+\bW_1\bR|}{|\bI+\bW_2\bR|}=C_1 (\bR)-C_2 (\bR)
\end{equation}
where negative $C(\bR)$ is interpreted as zero rate, $\bW_i =\bH_i' \bH_i$, and the secrecy capacity subject to the total Tx power constraint is
\begin{equation}
\label{eq.Cs}
C_s =\mathop {\max }\limits_{\bR\ge 0} C(\bR) \mbox{\ s.t.} \ tr\bR \le P_T
\end{equation}
where $P_T $ is the total transmit power (also the SNR since the noise is normalized). It is well-known that the problem in \eqref{eq.Cs} is not convex and hence very difficult to solve in general and explicit solutions for the optimal Tx covariance is not known for the general case, but only for some special cases, e.g. low/high SNR, MISO channels,  full-rank or rank-1 case \cite{Khisti-1}-\cite{Li} or for the parallel channel \cite{Khisti-08}\cite{Li-2}.

Since \eqref{eq.Cs} is not a convex problem in the general case, not only widely-used Karush-Kuhn-Tucker optimality conditions are not sufficient, but also the convergence of a numerical algorithm to a global optimum is very difficult if not impossible to insure since the standard tools of convex optimization fail to work and, in general, non-convex problems are much harder to deal with \cite{Boyd}. Thus, \eqref{eq.Cs} is very difficult to solve either analytically or numerically in the general case. Even when $C(R)$ is concave so that the problem becomes convex (when the channel is degraded, $\bW_1 \ge \bW_2$), its analytical solution is not known, except for the special cases noted above, and the known convex solvers \cite{Grant'06}-\cite{CVX} are not able to solve the problem, even in this convex setting so that a custom-made algorithm has to be developed.

To go around this difficulty, we use the following minimax representation of the secrecy capacity.

\section{Minimax Representation of Secrecy Capacity}
\label{sec:Max-Min}

A minimax representation of the secrecy capacity was obtained in  \cite{Khisti-2} via a channel enhancement argument and a clever bounding technique, which is instrumental for our algorithm and is summarized below.

\begin{thm}[Theorem 1 in \cite{Khisti-2}]
The secrecy capacity of Gaussian MIMO-WTC channel in \eqref{eq.C(R)} can be presented in the following minimax form :
\bal
\label{eq.Cs.R.K}
C_s=\max_{\bR}\min_{\bK} f(\bR,\bK) = \min_{\bK}\max_{\bR} f(\bR,\bK)
\eal
where
\bal
\label{eq.f(R,K)}
&f(\bR,\bK) =\frac{1}{2}\ln \frac{|\bI+\bK^{-1}\bH\bR\bH'|}{|\bI+\bW_2\bR|} \ge C(\bR),\\
\label{eq.K}
&\bK = \left(
        \begin{array}{cc}
            \bI & \bK_{21}' \\
            \bK_{21} & \bI \\
          \end{array}
        \right) \ge \bf{0}, \
\bH= \left(
      \begin{array}{cc}
        \bH_1 \\
        \bH_2 \\
      \end{array}
     \right),
\eal
and the optimization is over the set $\sS$ of all feasible $\bR, \bK$:
\bal
\label{eq.sS}
\sS = \{(\bR,\bK): tr\bR \le P,\ \bR, \bK \ge \bf{0},\ \bK \ \mbox{is as in \eqref{eq.K}} \}.
\eal
\end{thm}

The upper bound in \eqref{eq.f(R,K)} via $f(\bR,\bK)$ was obtained from a genie-aided receiver which knows $\by_2$ (in addition to $\by_1$) and $\bK$ represents noise covariance between $\bxi_1$ and $\bxi_2$. Minimization over $\bK$ is due to the fact that the true capacity does not depend on $\bK$ while the upper bound does so it's natural to seek the least upper bound. This bound can also be used in a numerical algorithm to evaluate the optimality gap with respect to $\min_{\bK}$ for each $\bR$. In fact, \eqref{eq.Cs.R.K} states that letting the receiver to know $\by_2$ in addition to $\by_1$ does not increase the secrecy capacity under the worst-case noise covariance, which is rather surprising.

\begin{remark}
2nd equality in \eqref{eq.Cs.R.K} expresses the saddle-point property, which is equivalent to the following inequalities (see e.g. \cite{Boyd}\cite{Zeidler}):
\bal
\label{eq.f(bR*,bK*)}
f(\bR,\bK^*) \le f(\bR^*,\bK^*) \le f(\bR^*,\bK)
\eal
which hold for any feasible $\bR, \bK$, where $(\bR^*,\bK^*)$ is the optimal (saddle) point of \eqref{eq.Cs.R.K}. These inequalities follow from von Neumann minimax Theorem since $f(\bK,\bR)$ is convex in $\bK$ for any fixed $\bR$ and concave in $\bR$ for any fixed $\bK$ (and for any channel, degraded or not), and the feasible set in \eqref{eq.sS} is convex.
\end{remark}
\begin{remark}
It is the convex-concave nature of $f(\bR,\bK)$ along with the saddle-point property in \eqref{eq.f(bR*,bK*)} and the constraints in \eqref{eq.sS} that make the respective KKT conditions sufficient for global optimality (see e.g. \cite{Boyd} and \cite{Ghosh} for more details; note that Slater's condition holds for these problems). This cannot be said about the original problem in \eqref{eq.Cs}. The sufficiency of the KKT conditions is the key for our algorithm and a proof of its convergence to a \textit{global} maximum (rather than just a stationary point).
\end{remark}

While the equivalence of \eqref{eq.Cs} and \eqref{eq.Cs.R.K} was established in \cite{Khisti-2}, an analytical solution of any one is not known in the general case. In fact, no analytical solution is known for the latter. Despite its more complicated appearance due to two conflicting optimizations, \eqref{eq.Cs.R.K} is in fact easier to solve than \eqref{eq.Cs}, at least numerically, since both optimizations are convex and the respective KKT conditions are sufficient for global optimality; a proof of convergence of the corresponding numerical algorithm to a \textit{global optimum} is also within reach for \textit{any} channel. While the standard tools developed for single convex optimization \cite{Boyd} do not apply directly here due to two conflicting optimizations involved, their primal-dual reformulation does work, as explained below.


We proceed to solve the minimax problem in \eqref{eq.Cs.R.K} via KKT conditions\footnote{See e.g. \cite{Boyd} for a background on KKT conditions.}. Subsequently, a numerical algorithm is developed with guaranteed convergence to a global optimum for any channel, degraded or not, which is not possible for \eqref{eq.Cs} due to its non-convex nature in the general case.  The Lagrangian for the problem in \eqref{eq.Cs.R.K} is
\bal \notag
L = f(\bR,\bK) &- tr \bM_1\bK + tr \bM_2\bR - \lambda (tr\bR-P)\\
 &+ tr \bLam(\bK-\bI)
\eal
where $\bM_1, \bM_2 \ge \bf{0}$ are (matrix) Lagrange multiplies responsible for the positive semi-definite constraints $\bK, \bR \ge \bf{0}$, $\lambda \ge 0$ is (scalar) Lagrange multiplier responsible for the total power constraint $tr\bR \le P$, and
\bal
\label{eq.A}
&\bLam = \left(
        \begin{array}{cc}
            \bLam_{1} & \bf{0} \\
            \bf{0}   & \bLam_2 \\
          \end{array}
        \right)
\eal
is a (matrix) Lagrange multiplier responsible for the constraint on $\bK$ as in \eqref{eq.K}. There are two sets of KKT conditions - one per optimization in \eqref{eq.Cs.R.K}. For the maximization over $\bR$, the KKT conditions are (to simplify notations, we have omitted the $\frac{1}{2}$ factor):
\bal \notag
\label{eq.KKT1}
&\nabla_R L = (\bI +\bW\bR)^{-1}\bW - (\bI+\bW_2\bR)^{-1}\bW_2
 + \bM_2 - \lambda\bI\\
 & \qquad= \bf{0}, \\
&\qquad \bM_2\bR = \bo,\\
& \qquad tr\bR \le P,\ \bR, \bM_2 \ge \bo, \ \lambda \ge 0,
\eal
where $\nabla_R$ is the gradient (derivative) with respect to $\bR$ and $\bW=\bH'\bK^{-1}\bH$. The KKT conditions for the minimization over $\bK$ are
\bal
\nabla_K L = (\bK+\bQ)^{-1} - \bK^{-1} - \bM_1 + \bLam = \bo, \\
\bM_1\bK = \bo, \\
\label{eq.KKT2}
\bK, \bM_1 \ge \bo,
\eal
and $\bK$, $\bLam$ are as in \eqref{eq.K}, \eqref{eq.A}; $\bQ=\bH\bR\bH'$. Here, we implicitly assume that $\bK > \bo$. While the singular case was treated in a separate way in \cite{Khisti-2}, we do not need a separate treatment here since our numerical algorithm is iterative and, at each step, it produces a non-singular $\bK$ which, however, may be arbitrary close to a singular matrix (i.e., may have arbitrary small but positive eigenvalues). This models numerically a case of singular $\bK$ and is a standard feature of the barrier method in general, where the boundary of the constraint set can be approached arbitrary closely but never achieved (see e.g. Chapter 11 in \cite{Boyd} for more detail). We remark that negligibly-small eigenvalues can be rounded off to 0 and they also imply that the numerical rank is low.

An optimal point in \eqref{eq.Cs.R.K} must satisfy both sets of KKT conditions simultaneously and these conditions are also sufficient for global optimality, as noted above. An analytical solution to these conditions is not known. Our numerical algorithm in Section V solves these two sets of KKT conditions in an iterative way, with guaranteed convergence to a globally-optimal point.

\section{Barrier Method for Minimax Optimization}

In this section, we first give a brief introduction into Newton and barrier methods for inequality-constrained optimization; the reader is referred to Chapters 9-11 of \cite{Boyd} for more details and background information. These two methods are used as key components to construct an algorithm for minimax optimization. Subsequently, this algorithm is adapted to the secrecy problem in \eqref{eq.Cs.R.K} and its guaranteed convergence to a global optimum is proved for any channel (degraded or not) in Section V.

\subsection{Minimax problem via primal-dual Newton method}

Newton method for an equality-constrained problem essentially transforms the problem into a sequence of quadratic problems for which the sufficient KKT conditions are a system of linear equations \cite{Boyd}.

Let us consider the minimax problem of the form\footnote{A similar problem, without equality constraints, have been briefly considered in \cite{Boyd}. More details can be found in \cite{Ghosh}. Our development here is tailored to be used for the secrecy problem in \eqref{eq.Cs.R.K}.}
\bal
\label{eq.max-min P}
\max_{\bx} \min_{\by} f(\bx,\by), \ \mbox{s.t.} \ \bA_x\bx=\bb_x, \bA_y\by=\bb_y
\eal
where vectors $\bx,\by$ represent optimization variables, the objective $f(\bx,\by)$ is concave in $\bx$ and convex in $\by$; given matrices $\bA_x, \bA_y$ and vectors $\bb_x, \bb_y$ represent the equality constraints for each variable. The KKT onditions for this problem are
\bal
\label{eq.KKT.xy}\notag
&\nabla_x f+\bA_x'\blam_x=0,\ \bA_x \bx - \bb_x = 0,\\
&\nabla_y f+\bA_y'\blam_y=0,\ \bA_y \by - \bb_y = 0,
\eal
where $\blam_x, \blam_y$ are dual variables, and they are sufficient for global optimality.

While the standard Newton method can be used for both optimizations, a proof of its convergence is challenging since the objective is not monotonous (it decreases in one step and increases at the other). The residual form of the Newton method is preferable since, as it was observed in \cite{Boyd}, it reduces the norm of the residual at each step and thus generates a monotonous sequence whose convergence to zero can be guaranteed. To introduce this method, let us aggregate variables, derivatives and parameters as follows:
\bal \notag
\bz &=
\left[
\begin{array}{cc}
  \bx\\ \by\\
\end{array}
\right],\
\blam =
\left[
\begin{array}{cc}
  \blam_x\\ \blam_y\\
\end{array}
\right],\
\bb = \left[
\begin{array}{cc}
  \bb_x\\ \bb_y\\
\end{array}
\right],\\
\bA &= \left[
\begin{array}{cc}
  \bA_x & \bf{0}\\
  \bf{0} & \bA_y\\
\end{array}
\right],
\eal
\bal
\nabla f =
\left[
\begin{array}{cc}
  \nabla_x f\\ \nabla_y f\\
\end{array}
\right],\
\nabla^2 f =
\left[
\begin{array}{cc}
  \nabla^2_{xx} f & \nabla^2_{xy} f \\
  \nabla^2_{yx} f & \nabla^2_{yy} f \\
\end{array}
\right],
\eal
The KKT conditions in \eqref{eq.KKT.xy} can be cast in a residual form:
\bal
\label{eq.KKT.z}
\br= [(\nabla f+\bA'\blam)', (\bA \bz - \bb)']'= \bo.
\eal
The Newton method iteratively solves $\br=\bo$ using 1st-order approximation (Newton step):
\bal
\notag
\br(\bw_0+\Delta\bw)&= \br(\bw_0)+D\br\Delta\bw + o(\Delta\bw)\\
\label{eq.r}
&\approx \br(\bw_0)+D\br\Delta\bw
\eal
where $\bw=[\bz',\blam']$ is the vector of aggregated (primal/dual) variables, $\bw_0$ and $\Delta\bw$ are its initial value and update,  $D\br$ is the derivative of $\br(\bw)$:
\bal
D\br = \left[\frac{\partial\br}{\partial\bz'}, \frac{\partial\br}{\partial\blam'}\right] =
\left[
\begin{array}{cc}
  \nabla^2 f(\bz_0) & \bA' \\
  \bA & \bf{0} \\
\end{array}
\right]=
\bT
\eal
and $\bT$ is the KKT matrix. Now, setting $\br(\bw_0+\Delta\bw)=\bo$ and solving for $\Delta\bw$ from \eqref{eq.r} gives the update
\bal
\label{eq.dw}
\Delta\bw:\  \bT\Delta\bw = -\br(\bw_0)
\eal
We further show in Section V that $\bT$ is non-singular for our problem so that this system of linear equations is guaranteed to have a unique solution for any set of parameters\footnote{While $\Delta\bw= -\bT^{-1}\br(\bw_0)$ is its analytical solution, it is not computed in practice since computing $\bT^{-1}$ is computationally-expensive and may result in loss of accuracy for ill-conditioned $\bT$, see e.g. \cite{Horn-1}.}.

Having the steps $\bdw = (\bdz', \bdlam')'$ computed,  the primal/dual variable updates are
\bal
\bz = \bz_0 + s\bdz, \blam = \blam_0 + s\bdlam
\eal
where the step size $s$ is found via the backtracking line search  \cite{Boyd} as in Algorithm 1 below.

\begin{algorithm}
\caption{Backtracking line search}
\label{alg.1}
\begin{algorithmic}
\Require $\bw_0,\ 0< \alpha<1/2,\ 0<\beta<1,\ s=1$.
\While {$|\br(\bw_0+s\bdw)| > (1-\alpha s)|\br(\bw_0)|\ $} {$\ s:=\beta s$}
\EndWhile
\end{algorithmic}
\end{algorithm}
In this Algorithm, $\alpha$ is a \% of the linear decrease in the residual one is prepared to accept at each step, and $\beta$ is a parameter controlling the reduction in step size at each iteration of the algorithm. The Newton method in combination with the backtracking line search  is guaranteed to reduce the residual norm $|\br(\bw)|$ at each step according to the following residual norm-reduction property \cite{Boyd}:
\bal
\label{eq.r.dt}
\frac{d}{ds}|\br(\bw_0+s\bdw)| = - |\br(\bw_0)| <0,
\eal
so that, for sufficiently small $s$, the residual indeed shrinks at each iteration (unless $|\br(\bw_0)| =0$, which implies that $\bw_0$ is optimal). This insures convergence of the algorithm to a global optimum since KKT conditions are sufficient for optimality and any locally-optimal point is automatically globally-optimal as the problem is convex.

\begin{algorithm}
\caption{Newton method for minimax optimization}
\label{alg.2}
\begin{algorithmic}
\Require $\bz_0, \blam_0, \alpha, \beta, \epsilon$
\Repeat
\State 1. Find $\bdz, \bdlam$ using Newton step in \eqref{eq.dw}.
\State 2. Find $s$ using the backtracking line search (Algorithm 1).
\State 3. Update variables: $\bz_{k+1}=\bz_k+s \bdz, \blam_{k+1}=\blam_k+s \bdlam$.
\Until $|\br(\bz_{k+1},\blam_{k+1})| \le \epsilon$.
\end{algorithmic}
\end{algorithm}

Based on this, the Newton method for minimax optimization is as in Algorithm 2. The convergence of this algorithm to a global optimum is insured by the convex/concave nature of the objective, sufficiency of the KKT conditions in \eqref{eq.KKT.xy}, non-singularity of the KKT matrix $\bT$ at each step  (as proved in Section V) and the norm-decreasing residual  property in \eqref{eq.r.dt}, which ensures that the method generates a sequence of sub-optimal solutions with monotonically decreasing residuals, for which the stationary point has zero residual and thus solves the sufficient KKT conditions. While the global optimum point corresponds to zero residual, $|\br| = 0$ (this is equivalent to the KKT conditions in \eqref{eq.KKT.xy}), the practical version $|\br| \le \epsilon$ of this condition is used in Algorithm 2 as a stopping criterion. This form of the stopping criteria is justified by not only the residual form $|\br| = 0$  of the KKT conditions, but also by the norm-decreasing residual property in \eqref{eq.r.dt}.

As a side remark, we note that this algorithm can also be used to solve the problem in \eqref{eq.max-min P} with $\max$ and $\min$ interchanged, due to the saddle point property.

\subsection{Barrier method for inequality-constrained problems}

Let us now combine the barrier method and the minimax method above to construct an algorithm for minimax optimization with equality and inequality constraints. Consider the following problem with inequality constraints:
\bal \notag
\label{eq.max-min ineq}
\max_{\bx} \min_{\by} f(\bx,\by), \ \mbox{s.t.} \ &\bA_x\bx=\bb_x,\ \bA_y\by=\bb_y,\\
&f_1(\bx) \le 0,\ f_2(\by) \le 0
\eal
where $f_1$ and $f_2$ are the constraint functions. The key idea of the barrier method is to use a soft instead of hard constraints by augmenting the objective with the barrier functions responsible for the inequality constraints so that the new objective for the problem in \eqref{eq.max-min ineq} becomes:
\bal
f_t(\bx,\by) = f(\bx,\by) + \psi_t(f_1(\bx)) - \psi_t(f_2(\by))
\eal
where we use the logarithmic barrier function:
\bal
\label{eq.psi_t}
\psi_t(x) = \frac{1}{t} \ln(-x)
\eal
and where $t$ is the barrier parameter. The barrier method transforms the inequality-constrained problem in \eqref{eq.max-min ineq} into the following problem without inequality constraints:
\bal
\label{eq.max-min B}
\max_{\bx} \min_{\by} f_t(\bx,\by), \ \mbox{s.t.} \ \bA_x\bx=\bb_x, \bA_y\by=\bb_y
\eal
The optimality gap due to this transformation can be upper bounded as follows.
\begin{prop}
The optimality gap of the barrier method in \eqref{eq.max-min B} applied  to the minimax problem in \eqref{eq.max-min ineq}  is as follows:
\bal
\label{eq.opt.gap}
|f(\bx^*(t),\by^*(t)) - p^*| \le 1/t
\eal
where $p^*$ is an optimal value of the original problem in \eqref{eq.max-min ineq} and $(\bx^*(t),\by^*(t))$ is an optimal point for the modified problem in \eqref{eq.max-min B}.
\end{prop}
\begin{proof}
This is a special case of Proposition 3 below with $m=n_1=n_2=1$.
\end{proof}

Thus, by selecting sufficiently high $t$, one can obtain arbitrary small gap. Newton method is used to solve the modified problem with any desired accuracy.

In practice, the modified problem is solved in an iterative way by selecting first a moderately-large value of $t$, solving the problem, increasing $t$ and using the previous solution as a starting point for a new one. In this way, the total number of Newton steps required to achieve certain accuracy is minimized \cite{Boyd}. The algorithm is as follows.

\begin{algorithm}
\caption{Barrier method}
\label{alg.3}
\begin{algorithmic}
\Require $\bz, \blam, \epsilon>0, t>0, \mu >1$
\Repeat
\State 1. Solve the problem in \eqref{eq.max-min B} using Newton method (Algorithm 2) starting at $\bz, \blam$.
\State 2. Update variables: $\bz := \bz^*(t), \blam := \blam^*(t),\ t := \mu t$.
\Until $1/t < \epsilon$.
\end{algorithmic}
\end{algorithm}

\section{Barrier Method for Secrecy Rate Maximization}

In this section, we use the minimax barrier method above to solve the optimal covariance problem in \eqref{eq.Cs.R.K} iteratively with guaranteed convergence to a global optimum, which is also optimal for \eqref{eq.Cs}.

\subsection{Choice of variables}

Since the original variables are positive semi-definite matrices $\bR, \bK$ and the barrier method above requires vectors, we have two options:

1. Use all entries of $\bR, \bK$ as independent variables via $\bx = vec(\bR),\ \by = vec(\bK)$, where operator $vec$ stacks all columns into a single vector. Enforce the symmetry constraints $\bR'=\bR,\ \bK'=\bK$ and the equality constraint on $\bK$ in \eqref{eq.K} via extra equality constraints.

2. Use only lower-triangular entries of $\bR$ as independent variables via $\bx=vech(\bR)$, where $vech$ stacks column-wise all lower-triangular entries into a single column vector, and use only $\bK_{21}$: $\by=vec(\bK_{21})$.

It can be shown that these two options are mathematically equivalent, i.e. produce exactly the same solutions at each step of Newton method. Option 2 is a preferable choice for implementation since the number of variables and constraints is reduced so that it is more efficient. Therefore, we use Option 2 for further exposition. Gradient and Hessian can be evaluated either numerically (in a standard way) or analytically as given below. We find the analytical evaluation to be preferable as numerical one entails a loss of precision while approaching an optimal point (this is especially pronounced at high SNR, large $t$ and for large systems).

Since the algorithm requires initial point to begin with, we use the following point:
\bal
\bR_0 = \frac{P}{m}\bI \rightarrow \bx_0 = vech(\bR_0),\\
\bK_0 = \bI \rightarrow \by_0 = \bo,\\
\blam_0 = \bo
\eal
As can be easily verified, the initial point above is feasible (i.e. satisfies the constraints). The choice of $\bR_0$ is motivated by the fact that isotropic signalling does not prefer any direction and thus is equally good a priori for any channel. $\bK_0$ corresponds to isotropic noise and is motivated by the same reason. It should be emphasized that the algorithm converges for any (feasible) initial point, due to the convex nature of the problem, to a global optimum; the difference is in how fast.

To account for the positive semi-definite constraints $\bR,\bK \ge \bo$, the following barrier function is used
\bal
\label{eq.psi_t(R)}
\psi_t(\bR) = \frac{1}{t} \ln|\bR|
\eal
so that the modified objective $f_t$ is
\bal
\label{eq.f_t(R,K)}
f_t(\bR,\bK) = f(\bR,\bK) + \psi_t(\bR) - \psi_t(\bK)
\eal
Note that this requires $\bK, \bR >\bo$, i.e. they are strictly inside of the feasible set but can approach the boundary arbitrary closely as $t$ increases, so that some eigenvalues may become arbitrary close to zero (and the numerical rank may be deficient); this models numerically the case of singular $\bR$ and/or $\bK$ and is a standard feature of the barrier method in general \cite{Boyd}. The inequality in \eqref{eq.opt.gap.2} makes sure that the optimality gap due to this can be made as small as desired. In a practical implementation, one can round off negligibly-small eigenvalues of $\bR$ to zero to simplify implementation.

After some manipulations (see Appendix for details), the gradients and Hessians can be expressed as:
\bal
\label{eq.nabla_f_t}
\nabla_x f_t = \bD_m' vec(\nabla_R f_t), \ \nabla_y f_t = \widetilde{\bD}_n' vec(\nabla_K f_t),
\eal
\bal \notag
\label{eq.nabla_xx_f_t}
\nabla_{xx}^2 f_t = - \bD_m'(\bZ_1\otimes\bZ_1 &- \bZ_2\otimes\bZ_2\\
 &+ t^{-1}\bR^{-1}\otimes\bR^{-1})\bD_m,
\eal
\bal \notag
\label{eq.nabla_yy_f_t}
\nabla_{yy}^2 f_t =  \widetilde{\bD}_n'(-(\bK &+\bQ)^{-1}\otimes (\bK+\bQ)^{-1}\\
 &+ (1+t^{-1}) \bK^{-1}\otimes \bK^{-1})\widetilde{\bD}_n,
\eal
\bal
\label{eq.nabla_xy_f_t}
\nabla_{xy}^2 f_t &= - \bD_m'(\bH'(\bK+\bQ)^{-1}\otimes \bH'(\bK+\bQ)^{-1})\widetilde{\bD}_n,
\eal
where
\bal
&\nabla_R f_t = \bZ_1 - \bZ_2 + t^{-1} \bR^{-1}, \\
&\nabla_K f_t = (\bK+\bQ)^{-1} - (1+t^{-1})\bK^{-1},\\
&\bZ_1 = (\bI+\bW\bR)^{-1}\bW,\\
&\bZ_2 = (\bI+\bW_2\bR)^{-1}\bW_2,
\eal
and $\otimes$ is a Kronecker product, $\bD_m$ is a $m^2\times m(m+1)/2$ duplication matrix defined from $vec(\bR)=\bD_m vech(\bR)$ \cite{Magnus}\cite{Harville}, $\widetilde{\bD}_n$ is a $n^2\times n_1n_2$ reduced duplication matrix defined from $d\bk = \widetilde{\bD}_n d\tilde{\bk}$, where
\bal \notag
&d\bk = vec(d\bK),\ d\tilde{\bk}=vec(d\bK_{21}), \\
&d\bK = \left(
        \begin{array}{cc}
            \bo & d\bK_{21}' \\
            d\bK_{21} & \bo \\
          \end{array}
        \right)
\eal
and $n=n_1+n_2$. It can be obtained from $\bD_n$ by removing its columns corresponding to all entries of $\bK$ but those in $\bK_{21}$.

It can be shown (see e.g. \cite{Loyka}) that using the full available power is optimal. Therefore, one can use the equality constraint $tr\bR = P$ instead of the inequality $tr\bR \le P$. The equality constraint matrix $\bA$ and vector $\bf{b}$ take the following form:
\bal
\label{eq.A.b}
\bA = [ \ba', \bo'],\ {\bf{b}}=P
\eal
where $\bI_m$ is $m\times m$ identity matrix, $\ba = vech(\bI_m)$, and $\bo$ is $n_1n_2\times 1$ zero vector, i.e. $\bA$ is a row vector and $\bb$ is a scalar in our setting.

With this choice of variables and initial points, Algorithm 3, in combinations with Algorithms 1 and 2, can now be used to solve numerically the minimax problem in \eqref{eq.Cs.R.K}.

\subsection{Convergence of the algorithm}

Here, we provide a proof of convergence of the proposed algorithm to a global optimum. First, one has to insure that Newton step is well defined for all $t, \bR, \bK >\bo$. This, in turn, insures that the Newton method produces a sequence of decreasing-norm residuals (according to \eqref{eq.r.dt}), which converge to zero for each $t$. Consequently, the minimax barrier method applied to our problem generates a sequence of sub-optimal points $\bz^*(t)$ that converges to a global optimum (a solution of the sufficient KKT conditions in \eqref{eq.KKT1}-\eqref{eq.KKT2}) as $t$ increases, since $f_t(\bR,\bK)$ is convex in $\bK$ and concave in $\bR$ and also twice continuously differentiable for each $\bR > \bo,\ \bK>\bo$ (more details can be found in \cite{Boyd}).

To make sure that Newton step is well defined for each $t, \bR, \bK >\bo$, we demonstrate that the KKT matrix for the modified objective $f_t$ is non-singular, so that the Newton equations have a well-defined solution as in \eqref{eq.dw}.

\begin{prop}
Consider the minimax problem in \eqref{eq.max-min P} for the objective in \eqref{eq.f_t(R,K)} under the equality constraint parameters as in \eqref{eq.A.b}. Its KKT matrix
\bal
\label{eq.KKT.matrix}
\bT=
\left[
\begin{array}{cc}
  \nabla^2 f_t & \bA' \\
  \bA & \bf{0} \\
\end{array}
\right]
\eal
is non-singular for each $t>0,\ \bR,\bK>\bo$.
\end{prop}
\begin{proof}
The proof is based on the following three Lemmas.

\begin{lemma}
The Hessian
\bal
\nabla^2 f_t = \breve{\bH} =
\left[
\begin{array}{cc}
  -\bH_{11} & \bH_{12}\\
  \bH_{21} & \bH_{22}\\
\end{array}
\right]
\eal
is non-singular if partial Hessians $\bH_{11}, \bH_{22}$ are non-singular, i.e. if $\bH_{11}, \bH_{22} > \bo $, where $\bH_{11}=-\nabla^2_{xx} f_t,\ \bH_{12}=\nabla^2_{xy} f_t,\ \bH_{21}=\bH_{12}'=\nabla^2_{yx} f_t,\ \bH_{22}=\nabla^2_{yy} f_t$. Furthermore, block $(1,1)$ $[\breve{\bH}^{-1}]_{11}$ of the inverse $\breve{\bH}^{-1}$ is also non-singular.
\end{lemma}
\begin{proof}
The proof is complicated by the fact that $\nabla^2 f_t$ is indefinite matrix, since $f_t$ is concave in $\bx$ and convex in $\by$ (i.e. $\nabla^2_{xx} f_t \le \bo,\ \nabla^2_{yy} f_t \ge \bo$), so that the standard proofs tailored for positive definite matrices \cite{Horn-1} do not apply here.
However, since $\bH_{11}, \bH_{22} > \bo$, it follows that
\bal \notag
\bS_{22} &= -\bH_{11} - \bH_{21}'\bH_{22}^{-1}\bH_{21} < \bo,\\  \bS_{11} &= \bH_{22} + \bH_{21}\bH_{11}^{-1}\bH_{21}' > \bo,
\eal
where $\bS_{11(22)}$ is Schur complement of $-\bH_{11}(\bH_{22})$, so that the matrix inversion Lemma in Proposition 2.8.7 of \cite{Bernstein} applies  and one can invert $\breve{\bH}$ as follows \footnote{This idea of the proof was suggested by a reviewer.}
\bal \notag
\breve{\bH}^{-1} &=
\left[
\begin{array}{cc}
  -\bH_{11} & \bH_{21}'\\
  \bH_{21} & \bH_{22}\\
\end{array}
\right]^{-1}\\
&= \left[
\begin{array}{cc}
   \bS_{22}^{-1} & -\bS_{22}^{-1} \bH_{21}'\bH_{22}^{-1}\\
  \bS_{11}^{-1}\bH_{21}\bH_{11}^{-1} & \bS_{11}^{-1}\\
\end{array}
\right]
\eal
which implies that $\breve{\bH}$ is non-singular and that $[\breve{\bH}^{-1}]_{11} = \bS_{22}^{-1} < \bo$.
\end{proof}

\begin{lemma}
The KKT matrix in Proposition 2 is non-singular under the conditions of Lemma 1.
\end{lemma}
\begin{proof}
We proceed as follows. Since the Hessian $\nabla^2 f_t = \breve{\bH}$ is non-singular (under conditions of Lemma 1), let us apply the following transformation that preserves the determinant of $\bT$:
\bal \notag
\widetilde{\bT} &=
\left[
\begin{array}{cc}
  \breve{\bH} & \bA' \\
  \bA & \bf{0} \\
\end{array}
\right]
\left[
\begin{array}{cc}
  \bI & -\breve{\bH}^{-1}\bA' \\
  \bo & \bI\\
\end{array}
\right]\\
&=
\left[
\begin{array}{cc}
  \breve{\bH} & \bo \\
  \bA & -\bA\breve{\bH}^{-1}\bA'\\
\end{array}
\right],
\eal
and observe that
\bal
\label{eq.Lemma.2.2}
|\widetilde{\bT}| = |\bT|=|\breve{\bH}|(-\bA\breve{\bH}^{-1}\bA')
\eal
(this follows from the properties of block-partitioned matrices and their determinants, see e.g. \cite{Zhang}). From Lemma 1, $|\breve{\bH}|\neq 0$. Further notice that $\bA\breve{\bH}^{-1}\bA' = \ba'[\breve{\bH}^{-1}]_{11}\ba < 0$, since $[\breve{\bH}^{-1}]_{11} < \bo$ from Lemma 1 and $\ba \neq 0$.
Using \eqref{eq.Lemma.2.2}, $|\bT|=|\breve{\bH}|(-\bA\breve{\bH}^{-1}\bA') \neq 0$ so that the KKT matrix $\bT$ is non-singular.
\end{proof}

Thus, Lemmas 1 and 2 establish the non-singularity of KKT matrix provided that partial Hessians $\nabla^2_{xx} f_t,\ \nabla^2_{yy} f_t$ are non-singular. This is indeed the case as Lemma 3 below shows.

\begin{lemma}
Partial Hessian $\nabla^2_{xx} f_t,\ \nabla^2_{yy} f_t$ in \eqref{eq.nabla_xx_f_t} and \eqref{eq.nabla_yy_f_t} are non-singular for each $t>0,\ \bR,\bK > \bo$.
\end{lemma}
\begin{proof}
See Appendix.
\end{proof}

Combining Lemmas 1-3, Proposition 2 follows.
\end{proof}

Thus, Proposition 2 insures that Newton step is always well-defined and hence generates a sequence of decreasing-norm residuals (according to \eqref{eq.r.dt}) which converges to zero for each $t>0$. The next proposition specifies the optimality gap of the minimax barrier method for a given $t$.

\begin{prop}
For each $t>0$, the optimality gap of the barrier method applied  to the minimax problem in \eqref{eq.Cs.R.K}  can be upper bounded as follows:
\bal
\label{eq.opt.gap.2}
|f(\bR^*(t),\bK^*(t)) - C_s| \le \max(m,n_1+n_2)/t
\eal
where $\bR^*(t),\bK^*(t)$ are the optimal signal and noise covariance matrices returned by the barrier method for a given $t$.
\end{prop}
\begin{proof}
Using the bounds for the minimax problem in \cite{Ghosh} and adopting them to the problem in \eqref{eq.Cs.R.K}, one obtains
\bal
\max_R f(\bR,\bK^*(t))-m/t &\le f(\bR^*(t),\bK^*(t))\\ \notag
 &\le \min_K f(\bR^*(t),\bK)+(n_1+n_2)/t
\eal
so that
\bal \notag
f(\bR^*(t),\bK^*(t)) &\le \min_K f(\bR^*(t),\bK)+(n_1+n_2)/t \\ \notag
&\le \max_R \min_K f(\bR,\bK)+(n_1+n_2)/t\\
&= C_s + (n_1+n_2)/t,\\
f(\bR^*(t),\bK^*(t)) &\ge \max_R f(\bR,\bK^*(t))-m/t\\ \notag
&\ge \min_K \max_R f(\bR,\bK)-m/t = C_s -m/t
\eal
from which \eqref{eq.opt.gap.2} follows.
\end{proof}

Therefore, using sufficiently large barrier parameter $t$ insures any desired accuracy, and $f(\bR^*(t),\bK^*(t)) \rightarrow C_s $ as $t \rightarrow \infty$. If desired accuracy is $\epsilon$, then the stopping criterion in Algorithm 3 should be $\max(m,n_1+n_2)/t < \epsilon$ (assuming that the Newton method produces sufficiently-accurate solution, which is always the case in practice due to its quadratic convergence, see \cite{Boyd}).

\subsection{Dual Problem}

While the algorithm above is designed to maximize the secrecy rate, its optimal covariance also solves the dual problem of minimizing the total transmit power subject to the secrecy rate constraint $C(\bR) \ge R_s$, i.e.
\bal
\label{eq.dual.problem}
\min\ tr\bR \ \mbox{\ s.t.}\ \ C(\bR) \ge R_s,\ \bR \ge 0
\eal
This can be easily shown by contradiction and observing that 1st inequality in \eqref{eq.dual.problem} always holds with equality, or by comparing the respective KKT conditions (which are necessary for optimality in both problems), both under the condition $R_s=C_s$.

\subsection{Per-antenna Power Constraints}
Different forms of power constraint can also be incorporated into the proposed algorithm in a straightforward way. In particular, the per-antenna power constraint $r_{ii} \le P_i$, where $r_{ii}$ is $i$-th diagonal entry of $\bR$ (power in antenna $i$) and $P_i$ is the maximum power of $i$-th antenna, can be adopted by eliminating matrix $\bA$ from the KKT equations and adding $m$ extra barrier terms $t^{-1} \ln (P_i - r_{ii})$ representing new power constraints in \eqref{eq.f_t(R,K)}. As a starting point, one can use e.g. $r_{ii}=P_i/2$.

In fact, these new constraints can be added to the existing ones as well, representing the scenario where not only the total power budget is limited but also the per-antenna powers are limited due to e.g. limited dynamic range of power amplifiers.

The convergence of this modified algorithm to a global optimum can be proved in the same way as above (with minor modifications). In particular, one can observe that the new barrier terms preserve the non-singularity of the KKT matrix and the convex nature of the problem.

\section{Degraded Channel}
If the channel is degraded, $\bW_1 \ge \bW_2$, then $C(\bR)$ is concave and the corresponding optimization problem in \eqref{eq.Cs} is convex. Therefore, the barrier method can be applied directly to this problem with guaranteed convergence to a global optimum. This reduces the problem complexity since there is no minimization over $\bK$ so that the number of variables reduces from $m(m+1)/2 +n_1 n_2$ to $m(m+1)/2$, which is a significant improvement when $n_1 n_2$ is large.

The modified objective (with the barrier term) becomes
\bal
f_t(\bR) = C(\bR) + \psi_t(\bR),
\eal
the variables are $\bz=\bx=vech(\bR)$ (no $\by$) and the equality constraint parameters are
\bal
\bA=\ba'=vech(\bI),\ \bb=P,
\eal

Non-singularity of the KKT matrix, which guarantees well-defined Newton steps, can be established  following the lines of the analysis in Section V. In particular, one observes that Lemmas 1-3 hold. Lemma 3 holds since
\bal
\nabla_{xx}^2 f_t < \bo
\eal
Lemma 1 holds since the Hessian in this case is $\breve{\bH}=\nabla_{xx}^2 f_t$. Lemma 2 holds since
\bal
\ba'\breve{\bH}^{-1}\ba < 0
\eal
so that the KKT matrix is non-singular and thus KKT conditions have a well-defined solution at each step of the barrier method.

The optimality gap in this case becomes
\bal
\label{eq.opt.gap.3}
|C(\bR^*(t)) - C_s| \le m/t
\eal
where $\bR^*(t)$ is an optimal $\bR$ returned by the Newton method for a given $t$, i.e. it is smaller for the same $t$ than in the non-degraded case \eqref{eq.opt.gap.2}, which is an extra advantage (in addition to having less variables). For desired accuracy $\epsilon$, the stopping criterion in Algorithm 3 is $m/t < \epsilon$.

As a side remark, we note that even though the problem is convex in this case, existing convex solvers (see e.g. \cite{Grant'06}-\cite{CVX}) cannot be used to solve it directly since they do not allow difference of logarithms or matrix powers in objective/constraint functions, while the algorithm above solves it with guaranteed convergence to a global optimum.

\section{Numerical Experiments}

To validate the algorithm and analysis and to demonstrate the performance of the algorithm, extensive numerical experiments have been carried out. Some of the representative results are shown below.

\begin{figure}[t]
\label{fig.1}
\centerline{\includegraphics[width=3.3in]{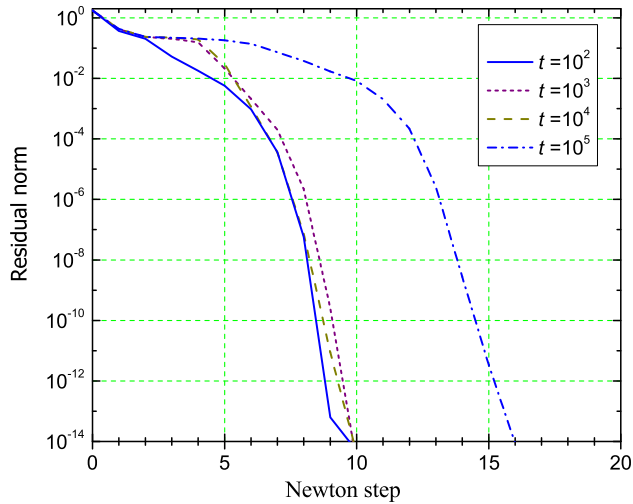}}
\caption{Convergence of the Newton method for different values of $t$; $m=2, P=10$, $\alpha=0.3, \beta=0.5$, $\bH_1, \bH_2$ as in \eqref{eq.H.example.1}. Note the presence of two convergence phases: linear and quadratic. It takes only about 10 to 20 Newton steps to reach the machine precision level.}
\end{figure}

\begin{figure}[t]
\label{fig.2}
\centerline{\includegraphics[width=3.3in]{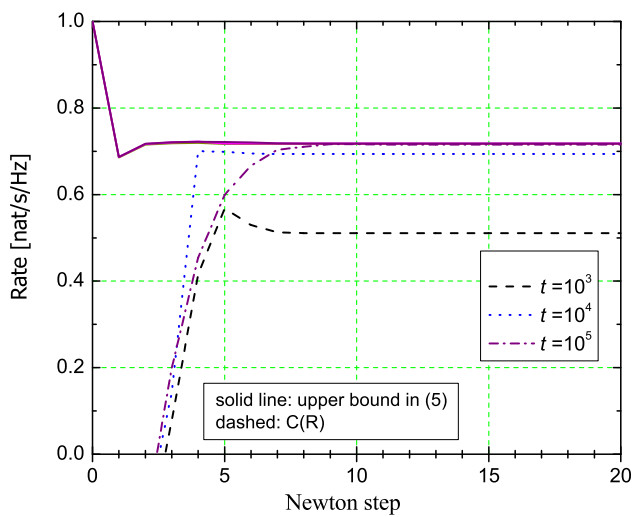}}
\caption{Secrecy rates for the same setting as in Fig. 2. Solid line - via the upper bound in \eqref{eq.f(R,K)} (the lines coincide for different $t$), dashed - via $C(\bR)$ in \eqref{eq.C(R)}.}
\end{figure}

Convergence of the Newton method for different values of the barrier parameter $t$ is demonstrated in Fig.2 for
\bal \notag
\label{eq.H.example.1}
\bH_1 &=
\left[
\begin{array}{cc}
    0.77 & -0.30\\
   -0.32 & -0.64\\
\end{array}
\right],\\
\bH_2 &=
\left[
\begin{array}{cc}
    0.54 & -0.11\\
   -0.93 & -1.71\\
\end{array}
\right],
\eal
which shows the residual $\br$ Euclidian norm versus Newton steps. Even though this channel is not degraded, since the eigenvalues of $\bW_1-\bW_2$ are $\{0.395, -3.293\}$, the algorithm does find the global optimum (this particular channel was selected because it is "difficult" for optimization). Note the presence of two convergence phases: linear and quadratic, which is typical for Newton method in general. After the quadratic phase is reached, the convergence is very fast (water-fall region). It takes about 10-20 Newton steps to reach very low residual (at the level of machine precision). This is in agreement with the observations in \cite{Boyd} (although obtained for different problems).

\begin{figure}[t]
\label{fig.3}
\centerline{\includegraphics[width=3.3in]{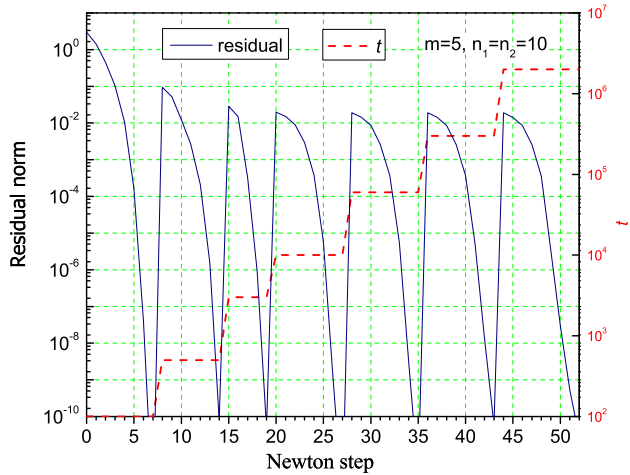}}
\caption{Convergence of the barrier method (incrementally increasing $t$); $m=5, n_1=n_2=10, P=10$, $\alpha=0.3, \beta=0.5$, $\mu=5$, $\bH_1, \bH_2$ are randomly generated (i.i.d. Gaussian entries of zero mean and unit variance). It takes about 5 to 10 steps to reduce the residual to a very low value of $10^{-10}$ for each $t$. }
\end{figure}

\begin{figure}[t]
\label{fig.4}
\centerline{\includegraphics[width=3.3in]{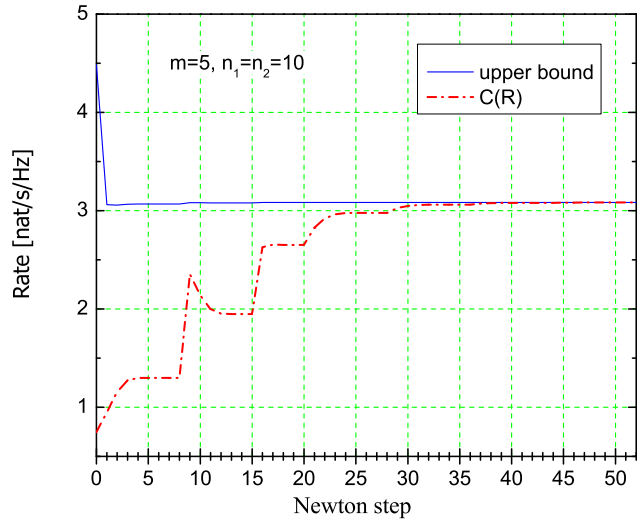}}
\caption{Secrecy rates for the same setting as in Fig. 4. Solid line - via the upper bound in \eqref{eq.f(R,K)}, dashed - via $C(\bR)$ in \eqref{eq.C(R)}. Note that while the capacity value evaluated via the upper bound converges very fast, significantly more iterations are required for convergence of the secrecy rate $C(\bR)$. We attribute this to the fact that $C(\bR)$ is more sensitive to $\bR$ than $f(\bR,\bK)$ is. Also note the significantly non-monotonique behavior of the former. }
\end{figure}

Fig. 3 shows the corresponding secrecy rate evaluated via the upper bound in \eqref{eq.f(R,K)} and the actual achievable rate via $C(\bR(t))$ in \eqref{eq.C(R)}, where $\bR(t)$ is an optimal covariance  at a particular step of the Newton method and for a given $t$. As the algorithm converges, they become almost equal if $t$ is sufficiently large (in this case, about $10^4...10^5$). While $t$ has negligible impact on the upper bound, it does affect significantly the corresponding $C(\bR(t))$ (since the optimal covariance $\bR(t)$ returned by the barrier method depends on $t$ and $C(\bR)$ is sensitive to $\bR$), so that the choice of $t$ is not critical if the secrecy capacity is the only quantity of interest (since the upper bound is quite tight even for moderate $t$). However, if a transmitter is implemented with the optimal covariance $\bR(t)$ returned by the algorithm, it is $C(\bR(t))$ that determines the achievable rate and this choice is important. We attribute this fact to higher sensitivity of $C(\bR)$ to $\bR$ compared to that of $f(\bR,\bK)$. Similar observations apply to the number of Newton steps required to achieve a certain performance: if $C_s$ is the quantity of interest, the upper bound converges to it in about 3-5 steps. However, if implementing $\bR$ is involved, one should use $C(\bR)$ and, in addition to proper choice of $t$, it takes about 5...10 steps to achieve the convergence. Note that, in both cases, the number of steps is not large and the execution  time is small (a few seconds). In general, larger $t$ and $m, n_1, n_2$ require more steps to achieve the same accuracy. As expected, the behavior of upper bound is not monotonic while the residual norm does decrease monotonically in each step.

Fig. 4 and 5 demonstrate the convergence of the minimax barrier method (incrementally increasing $t$) for a larger system ($m=5$, $n_1=n_2=10$). Note that a very low residual value of $10^{-10}$ is achieved after about 7 Newton steps for each value of $t$. Using incrementally-increasing $t$ as opposed to a fixed large value results in a smaller number of the total Newton steps required to achieve a given residual value and is less sensitive to system parameters and size. Also observe from Fig. 5 that while the upper bound converges quite fast (in a few Newton steps), it takes significantly more steps for $C(\bR)$ to converge and the convergence process is significantly non-monotonic.

To demonstrate the convergence performance for different channel realizations, Fig. 6 and 7 show the distribution (histograms) of the number of steps required to achieve the residual of $10^{-10}$ and $10^{-8}$ for 100 randomly-generated channels (with i.i.d. Gaussian entries of zero mean and unit variance) for $m=4, n_1=n_2=3$ and $m=5, n_1=n_2=10$ systems. While the actual number of required steps depends on a particular channel realization, 20 to 40 steps are sufficient in most cases. We attribute this to the two-phase behaviour of the algorithm's convergence: once the quadratic (water-fall) phase is reached, it takes just a few steps to reduce the residual to a very low value (which is consistent with similar observations in \cite{Boyd}, albeit for different problems). Different channel realizations result in a different number of required steps for the linear phase, before the quadratic phase is reached, but do not affect much the latter.

\begin{figure}[t]
\label{fig.5}
\centerline{\includegraphics[width=3.4in]{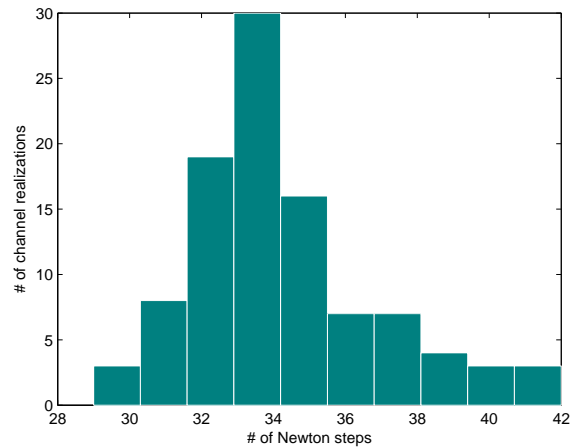}}
\caption{ A histogram showing the distribution of the number of Newton steps needed to achieve the residual of $10^{-10}$ via the minimax barrier method for 100 randomly generated channels (i.i.d. Gaussian entries of zero mean and unit variance); $P=10$, $\alpha=0.3, \beta=0.5$, $m=4, n_1=n_2=3$, $t_0=100,\ t_{max}=10^5$, $\mu=10$. }
\end{figure}

\begin{figure}[t]
\label{fig.6}
\centerline{\includegraphics[width=3.44in]{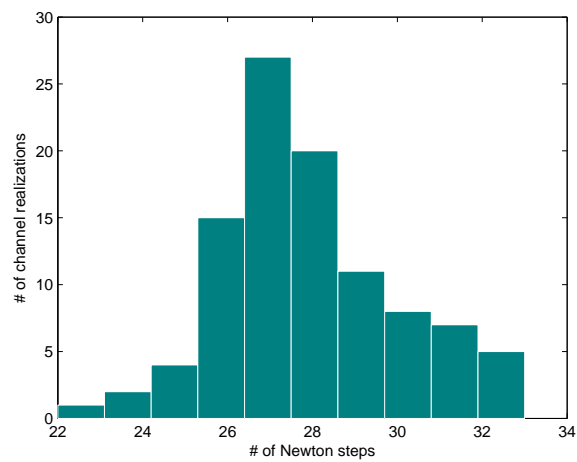}}
\caption{ A histogram showing the distribution of the number of Newton steps needed to achieve the residual of $10^{-8}$ for 100 randomly generated channels (i.i.d. Gaussian entries of zero mean and unit variance); $m=5, n_1=n_2=10$, $t_0=100,\ t_{max}=10^5$, $\mu=10$, $P=10$, $\alpha=0.3, \beta=0.5$. }
\end{figure}

\section{Conclusion}

Global secrecy rate maximization for (non-degraded) Gaussian MIMO-WTC has been discussed. The problem is challenging due to its non-convex nature and no analytical solution is known for this setting. While the known numerical algorithms converge to a stationary point (which may be a local rather than global maximum or just a saddle point), the algorithm proposed herein is guaranteed to converge to a \textit{global} rather than \textit{local} maximum. The algorithm is based on the minimax reformulation of the secrecy capacity problem (to insure global convergence) and the primal-dual reformulation of the Newton method in combination with the barrier method. A proof of its global convergence is also given. Numerical experiments indicate that 20 to 40 Newton steps are sufficient for convergence with high precision (up to the machine precision level). Extra power constraints (e.g. maximum per-antenna power) can be easily incorporated in the algorithm. The dual problem of total power minimization subject to the secrecy rate constraint can also be solved.

\section{Acknowledgement}
The authors would like to thank A.Khisti for numerous stimulating and insightful discussions, and the reviewers for constructive comments and suggestions.

\section{Appendix}

\subsection{Gradients and Hessians}

To derive the gradient and Hessian expressions, we use the tools of matrix differential calculus \cite{Magnus}\cite{Harville}. Let us consider $f(\bX)=\ln|\bX|$, where $\bX>\bo$ is $n\times n$ positive definite matrix. Using the perturbation method,
\bal
\label{eq.A1.1}
f(&\bX +d\bX)= \ln|\bX|+\ln|\bI+\bX^{-1}d\bX|\\ \notag
&= f(\bX) +\sum_i \lambda_i(\bX^{-1}d\bX) -\frac{1}{2}\lambda_i^2(\bX^{-1}d\bX)  + o(\lambda_i^2)\\ \notag
&= f(\bX)+tr(\bX^{-1}d\bX) -\frac{1}{2} tr(\bX^{-1}d\bX\bX^{-1}d\bX) + o(\{\lambda_i^2\})
\eal
Using
\bal \notag
tr(\bX^{-1}d\bX) &= vec(d\bX)'vec(\bX^{-1})\\
&=d\bx'\bD_n'vec(\bX^{-1})
\eal
where $d\bx=vech(d\bX)$, one obtains the gradient $\nabla_x f = \bD_n'vec(\bX^{-1})$. Applying this to
\bal
f(\bK)=\ln|\bI+\bK^{-1}\bQ|=\ln|\bK+\bQ|-\ln|\bK|,
\eal
$\nabla_y f_t$ follows. Using
\bal \notag
tr(\bX^{-1}d\bX\bX^{-1}d\bX) &= vec(d\bX)'(\bX^{-1}\otimes\bX^{-1})vec(d\bX)\\
 &= d\bx'D_n'(\bX^{-1}\otimes\bX^{-1})\bD_nd\bx
\eal
the Hessian $\nabla_{xx}^2 f$ can be identified as
\bal
\nabla_{xx}^2 f = - \bD_n'(\bX^{-1}\otimes\bX^{-1})\bD_n
\eal
Applying this to $f(\bK)$, $\nabla_y^2 f_t$ follows.

To derive $\nabla_x f_t$ and $\nabla_x^2 f_t$, use a modification of \eqref{eq.A1.1} for $f(\bR)=\ln|\bI+\bW\bR|$:
\bal \notag
\label{eq.A1.6}
f(\bR+d\bR)= f(\bR) &+tr(\bZ d\bR) -\frac{1}{2} tr(\bZ d\bR \bZ d\bR)\\
&+ \sum_i o(\lambda_i^2)
\eal
where $\bZ=(\bI+\bW\bR)^{-1}\bW$, so that
\bal
\nabla_r f_t &= \bD_m' vec(\bZ)
\eal
where $\br = vech(\bR)$, and
\bal
\nabla_r^2 f_t &= -\bD_m' (\bZ\otimes \bZ) \bD_m
\eal
from which \eqref{eq.nabla_f_t}, \eqref{eq.nabla_xx_f_t} follow, where we have used the following identities \cite{Magnus}:
\bal \notag
&tr(\bA\bB) = vec(\bA')'vec(\bB),\\
 &tr(\bA\bB\bC\bD)= (vec\bD)'(\bA\otimes\bC') vec(\bB')
\eal
and the fact that $\bZ$ is symmetric, $\bZ'=\bZ$.
To derive $\nabla_{xy}^2 f_t$, observe that
\bal
\nabla_{kr}^2 f(\bR,\bK) = \nabla_{kr}^2 \ln|\bK+\bH\bR\bH'|
\eal
where $d\bk=vec(d\bK)$, so that one needs to consider only
\bal
\tilde{f}(\bR,\bK) = \ln|\bK+\bH\bR\bH'|
\eal
for which the perturbation method gives
\bal \notag
\tilde{f}(\bR &+d\bR,\bK+d\bK) = \tilde{f}(\bR,\bK)\\
 &- tr(\bH'(\bK+\bQ)^{-1} d\bK (\bK+\bQ)^{-1}\bH d\bR)
+ \Delta\tilde{f}
\eal
where $\Delta\tilde{f}$ denotes all other terms (which do not affect the mixed derivatives), from which \eqref{eq.nabla_xy_f_t} follows by using $vec$ operator inside the trace.

\subsection{Proof of Lemma 3}

Observe that $\bQ \ge \bo$ so that $(\bK+\bQ)^{-1} \le \bK^{-1}$ and thus
\bal
\bK^{-1}\otimes \bK^{-1}-(\bK+\bQ)^{-1}\otimes (\bK+\bQ)^{-1} \ge \bo
\eal
(this follows from the properties of Kronecker products, see e.g. \cite{Zhang}) and
\bal \notag
(1+t^{-1}) \bK^{-1}\otimes \bK^{-1} &-(\bK+\bQ)^{-1}\otimes (\bK+\bQ)^{-1}\\
& \ge t^{-1}\bK^{-1}\otimes \bK^{-1} > \bo
\eal
Now consider the following quadratic form for any $\by\neq \bo$:
\bal
\by'\nabla_{yy}^2 f_t \by =  \tilde{\by}'((1 &+ t^{-1}) \bK^{-1}\otimes \bK^{-1}\\  \notag
& - (\bK+\bQ)^{-1}\otimes (\bK+\bQ)^{-1})\tilde{\by}
> 0
\eal
since $\tilde{\by} = \widetilde{\bD}_n\by \neq 0$ (this follows from the fact that all columns of $\widetilde{\bD}_n$ are linearly independent, which in turn is implied by linear independence of columns of $\bD_n$ since it has a full column rank \cite{Magnus}). Therefore, $\nabla_{yy}^2 f_t > \bo$.
Non-singularity of $\nabla_{xx}^2 f_t$ can be proved in a similar way. First, one observes that $\bW \ge \bW_2$:
\bal
\bW &= \bH'\bK^{-1}\bH \\
&= [\bH_1' \bH_2']
\left[
\begin{array}{cc}
   \bI & \bK_{21}' \\
   \bK_{21} & \bI \\
\end{array}
\right]^{-1}
\left[
\begin{array}{cc}
   \bH_1 \\
   \bH_2 \\
\end{array}
\right] \\  \notag
\label{eq.Lemma3.4.3}
&= \bH_2'\bH_2 + (\bH_1-\bK_{21}'\bH_2)'(\bI-\bK_{21}'\bK_{21})^{-1}\\
&\qquad\qquad\qquad\times (\bH_1-\bK_{21}'\bH_2)\\
&\ge \bH_2'\bH_2 = \bW_2
\eal
since 2nd term in \eqref{eq.Lemma3.4.3} is positive semi-definite, where we have used the matrix inversion Lemma:
\bal
\bK^{-1} &=
\left[
\begin{array}{cc}
  \bI & \bK_{21}'\\
  \bK_{21} & \bI\\
\end{array}
\right]^{-1}\\ \notag
&=
\left[
\begin{array}{cc}
  (\bI-\bK_{21}'\bK_{21})^{-1} & \bK_{21}' (\bK_{21}\bK_{21}'-\bI)^{-1} \\
  (\bK_{21}\bK_{21}'-\bI)^{-1}\bK_{21} & (\bI-\bK_{21}\bK_{21}')^{-1}\\
\end{array}
\right]
\eal
and the fact that $\bK_{21}'\bK_{21} <\bI$,  $\bK_{21}\bK_{21}' <\bI$, which follows from $\bK>\bo$ (since this implies $|\bK_{21}|_2 <1$, where $| \cdot |_2$ is the spectral norm, see e.g. \cite{Zhang}). Therefore, $\bZ_1 \ge \bZ_2$, which follows from the following argument when $\bW,\ \bW_2$ are non-singular:
\bal
\bW \ge \bW_2 &\Rightarrow \bW^{-1} \le \bW_2^{-1}\\
&\Rightarrow \bW^{-1}+\bR \le \bW_2^{-1}+\bR \\  \notag
&\Rightarrow \bZ_1=(\bW^{-1}+\bR)^{-1}\\
 & \qquad \ \ \ge (\bW_2^{-1}+\bR)^{-1}=\bZ_2
\eal
When $\bW$ and/or $\bW_2$ are singular, one can use the continuity argument \cite{Zhang}: use $\bW_{\epsilon}=\bW+\epsilon\bI >\bo$, $\bW_{2\epsilon}=\bW_2+\epsilon\bI >\bo$ with $\epsilon>0$, instead of $\bW,\ \bW_2$ and then take $\epsilon \rightarrow 0$; since both sides of the inequality are continuous functions, the result follows. Since $\bZ_1 \ge \bZ_2$, it follows that  $\bZ_1\otimes\bZ_1 \ge \bZ_2\otimes\bZ_2$ and thus
\bal
\bZ_1\otimes\bZ_1 - \bZ_2\otimes\bZ_2 + t^{-1}\bR^{-1}\otimes\bR^{-1} > \bo
\eal
(since $\bR^{-1}\otimes\bR^{-1} > \bo$) from which it follows that $\nabla_{xx}^2 f_t < \bo$.


%

\end{document}